\documentclass[pra,twocolumn]{revtex4-1} 
\usepackage{amssymb,amsmath,bm} 
\usepackage{mathrsfs}
\usepackage{cases,color}

\makeatletter

\def \v#1{{\bm #1}}
\def \be {\begin{equation}}
\def \ee {\end{equation}}
\newcommand{\Exp}[1]{\,\mathrm{e}^{\mbox{\footnotesize$#1$}}}
\newcommand{\I}{\mathrm{i}}

\newcommand{\tr}[1]{\mathrm{tr}\left(#1\right)}
\newcommand{\Tr}[1]{\mathrm{Tr}\left\{#1\right\}}

\def \Re{\mathrm{Re}\,}
\def \Im{\mathrm{Im}\,}

\newcommand{\sldin}[2]{\langle#1,#2\rangle_{\rho_{\theta}}^{S}}
\newcommand{\rldin}[2]{\langle#1,#2\rangle_{\rho_{\theta}}^{R}}
\def \del{\partial}

\def \lra{\Leftrightarrow}

\def \cD{{\cal D}}

\def \cM{{\cal M}}
\def \cP{{\cal P}}
\def \cH{{\cal H}}
\def \cX{{\cal X}}
\def \sofh{{\cal S}({\cal H})}
\def \barsofh{\overline{{\cal S}}({\cal H})}

\def \bbr{{\mathbb R}}
\def \bbc{{\mathbb C}}

\def \sofc2{{\cal S}({\mathbb C}^2)}

\def \lofh{{\cal L}({\cal H})}
\def \lofhh{{\cal L}_h({\cal H})}
\def \Dop#1{{\cal D}_{\rho_\theta}(#1)}

\def \cK{{\cal K}}

\newcommand \rspan[1]{\mathrm{span}_\bbr\{#1\}}
\newcommand \cspan[1]{\mathrm{span}_\bbc\{#1\}}

\def \QmetrologyRef{glm11,ta14,ddjk15,ps14,sbd16,drc17}
\def \CstatRef{rao,lc98,kiefer,wasserman,pazman}
\def \QlanRef{HM08,GK06,KG09,YFG13,YF17,YCH18}

\newtheorem{theorem}{Theorem}[section]
\newtheorem{lemma}[theorem]{Lemma}

\newtheorem{corollary}[theorem]{Corollary}
\newtheorem{definition}[theorem]{Definition}

\newenvironment{proof}[1][Proof:]{\begin{trivlist}
\item[\hskip \labelsep {\bfseries #1}]}{\end{trivlist}}

\newcommand{\qed}{\nobreak \ifvmode \relax \else
      \ifdim\lastskip<1.5em \hskip-\lastskip
      \hskip1.5em plus0em minus0.5em \fi \nobreak
      \vrule height0.75em width0.5em depth0.25em\fi}
      
\definecolor{dgreen}{RGB}{0,150,0}
\newcommand{\blue}[1]{\textcolor{blue}{#1}}
\newcommand{\red}[1]{\textcolor{red}{#1}}
\newcommand{\dgreen}[1]{\textcolor{dgreen}{#1}}
\definecolor{mygray}{gray}{0.6}

\begin{document}
\title{Classification and characterization of quantum parametric models in quantum estimation theory}
\author{Jun Suzuki\\
junsuzuki@uec.ac.jp}
\date{\today}
\affiliation{
Graduate School of Informatics and Engineering, The University of Electro-Communications,\\
1-5-1 Chofugaoka, Chofu-shi, Tokyo, 182-8585 Japan
}

\begin{abstract}
In this paper, we characterize quantum parametric models into different classes 
based on the estimation error bound, known as the Holevo bound. 
These classes are given by the classical, quasi-classical, D-invariant, and asymptotically classical models. 
We first explore the relationships among these four models and show that: 
i) The classical model having the diagonal elements only is characterized by 
the intersection of the D-invariant and asymptotically classical models. 
ii) There exists a gap between the classical model and the quasi-classical model, 
where all logarithmic derivative operators commute with each other. 
Further, we characterize each class with several equivalent conditions. 
This result then reveals the geometrical understanding of quantum statistical models. 
\end{abstract}


\maketitle

\section{Introduction}
Model classification is an important subject in practice and has been studied extensively in statistics. 
For example, if we know that obtained data are drawn according to a particular good class 
of statistical models, well-established estimation methods for this class can be applied to 
make estimates on the statistical model. 
There are, of course, bad statistical models in the sense that 
it is extremely hard to make any statistical estimates even numerically. 
In classical statistics, there exists a variety of different parametric models studied in details.
As a concrete example of a good statistical model, if we know that data is described by 
the standard linear response model with an equal variance, 
we can immediately apply the best linear unbiased estimator, which can be 
computed analytically. In reality, experimental data are affected by 
many unknown factors and considerable amounts of efforts haven devoted to study the 
general non-linear response models in statistics, see for example Refs.~\cite{\CstatRef}. 

Information geometry offers a completely different motivation to 
the model classification problem based on the geometrical properties of parametric models \cite{ANbook}. 
The most famous model is the exponential family, or the log-linear model, defined as an auto-parallel 
sub-manifold with respect to the exponential connection. What is remarkable regarding 
the exponential family is that achievability of the 
Cram\'er-Rao (CR) bound for the finite sample is given as 
if and only if the parametric model is the exponential 
family and the parameter to be estimated is an m-affine parameter of the model. 

The non-commutative extension of classical statistics to a quantum system 
was initiated in 1960s by Helstrom \cite{helstrom} and has been one of the fundamental 
problems in quantum information theory until today. 
The point-estimation problem about quantum states is 
the fundamental problem in theory and is also important for practical applications. 
In particular, recent advances in quantum metrology, quantum sensing, and quantum imaging, 
i.e., high precision measurement methods utilizing quantum resources, 
has triggered many activities in the field, see reviews on these subjects \cite{\QmetrologyRef}. 
Despite these efforts in past, there exist many open problems regarding 
multi-parameter estimation problems. One such fundamental problem is 
an explicit expression for the optimal estimation strategy that sets 
the bound for the estimation error. 
In classical statistics, this estimation error is bounded by 
the well-known CR bound, and an optimal estimator is 
the maximum likelihood estimator that asymptotically achieves this bound. 
Importantly, the CR bound in statistics is analytically calculated by the classical Fisher information matrix 
of the statistical model. 
A quantum version of this result is still missing mainly due to a nontrivial optimization 
for the measurement degrees of freedom, and also partially due to the fact that 
there exist many quantum versions of the Fisher information in the quantum system. 
In particular, quantum CR bounds, which are defined by quantum Fisher information matrices, 
cannot be achieved even asymptotically in general. 
A unified understanding on this fundamental estimation error bound is given by the Holevo bound \cite{holevo}. 
Unlike as in the classical case, this bound is expressed as a non-linear optimization problem. 

Model classification for the quantum parametric models is also an important problem, 
but it seems that this problem is less attracted so far by the quantum information community. 
The first attempt of model classification for the quantum estimation theory 
was due to Holevo \cite{holevo}. He introduced a particular class of quantum statistical models, 
called a {\it D-invariant} model, and showed that the right logarithmic derivative (RLD) CR 
bound can be achieved by the D-invariant model. Another non-trivial extension of 
model classification was studied by Nagaoka \cite{nagaoka89-2}. He defined the quantum exponential 
family, and showed that the symmetric logarithmic derivative (SLD) CR bound 
can be achieved uniformly by the quantum exponential family. 

In this paper, we make an attempt at classifying quantum parametric models based on the 
ultimate precision bound, i.e. the Holevo bound \cite{holevo}. One of the advantages of this approach is 
that we can immediately write down the achievable precision bound if a given model belongs to 
classes of models studied in this paper. The current paper is based on the results presented 
in Ref.~\cite{jsJMP}, where we analyzed the structure of the Holevo bound in detail for a qubit system. 
We derived an explicit formula for any qubit model together with characterization of special classes of the qubit models. 
We also classified the D-invariant model for the general qudit model together with non-trivial characterization of this model. 
In this paper, we continue to explore possible classification of quantum parametric models into several classes in which 
the Holevo bound can be expressed in closed formulas. We are also motivated by analyzing the structure of 
the tangent space and several quantum metrics on the quantum-state space. 

In this paper, we consider four different classes: 
The first class is the classical model where a quantum statistical model is reduced to a parametric model in classical statistics.  
This is because quantum statistical models defined by a set of positive semi-definite matrices with a unit trace 
contain classical statistical models as a special case. 
More precisely, when a given family of quantum states is simultaneously diagonalizable 
for all parameter values, the problem at hand can be reduced to the one in classical statistics. 
Though, this definition is trivial, it is important to characterize such the classical statistical model 
as properties of the tangent space. This is because the local property plays an important role in 
the quantum estimation theory in general. 
For the classical model, it is easy to see that the Holevo bound is simply reduced 
to the form of the classical Fisher information 
computed by the eigenvalues of the given quantum state. 

The second class is known as the quasi-classical model defined by the condition 
imposing all symmetric logarithmic derivative (SLD) operators commute with each other. 
When this condition is satisfied, it is clear that we can construct an optimal measurement 
by diagonalizing all SLD operators simultaneously. This then achieves the SLD CR bound 
for any finite sample size. As indicated by the name, this class is still quantum 
and is different from the classical model in general. 

The third class is known as the D-invariant model introduced by Holevo \cite{holevo}. 
It was shown in Ref.~\cite{jsJMP} that the Holevo bound is equivalent 
to the right logarithmic derivative (RLD) CR bound if and only if the model is D-invariant. 
In this paper, we put further step into characterizing the D-invariant model. 

The fourth class is when the Holevo bound coincides with the symmetric logarithmic derivative (SLD) CR bound. 
We call this class of models as the asymptotically classical in the sense that the model is asymptotically 
equivalent to a classical gaussian model in the local asymptotic normality (LAN) theory \cite{\QlanRef}. 
We note that the asymptotically classical model was introduced and analyzed in Refs.~\cite{jsJMP,jsqit32}. 
In Ref.~\cite{RJDD16}, the authors independently investigated the same problem and they called 
this condition as the compatibility condition. In this paper we give a more detailed analysis on 
their compatibility condition and derive several equivalent characterization of the asymptotically classical model. 

The aim of this paper is not just to classify quantum models into classes mentioned above, 
but to derive several equivalent conditions characterizing each class for the parametric family of quantum states. 
The results are given by theorems in Sec.~\ref{sec4}. We further examine relations among these classes. 
In Fig.~\ref{fig1}, we summarize the relations among four different classes of quantum statistical models. 
Figure \ref{fig2} in Sec.~\ref{sec4} also represents a schematic diagram for one of the main results of this paper. 

\begin{figure}
\setlength{\unitlength}{1mm}
\begin{picture}(80,45)(-5,0)
\linethickness{0.4pt}
\put(0,5){\line(1,0){70}}
\put(0,45){\line(1,0){70}}
\put(70,5){\line(0,1){40}}
\put(0,5){\line(0,1){40}}
\put(24,25){\blue{\oval(30,30)}}
\put(45,25.3){\red{\oval(35,24.3)}}
\put(39.3,25.3){\dgreen{\oval(23.5,24.3)}}
\put(5,40){\makebox(0,0)[c]{{\Large$\cM$}}}
\put(22,23.5){\makebox(0,0)[c]{\blue{$\cM_{D}$}}}
\put(32.5,23.5){\makebox(0,0)[c]{$\cM_{C}$}}
\put(43.5,23.5){\makebox(0,0)[c]{\dgreen{$\cM_{QC}$}}}
\put(57,23.5){\makebox(0,0)[c]{\red{$\cM_{AC}$}}}
\linethickness{0.06pt}
\put(12,11.3){\blue{\line(0,1){27.5}}}
\put(15,10){\blue{\line(0,1){30}}}
\put(18,10){\blue{\line(0,1){30}}}
\put(21,10){\blue{\line(0,1){30}}}
\put(24,10){\blue{\line(0,1){30}}}
\put(27,10){\blue{\line(0,1){30}}}
\put(30,10){\blue{\line(0,1){30}}}
\put(33,10){\blue{\line(0,1){30}}}
\put(36,11.3){\blue{\line(0,1){27.5}}}
\put(28,17.5){\red{\line(1,0){34}}}
\put(27.5,21.5){\red{\line(1,0){35}}}
\put(27.5,25.5){\red{\line(1,0){35}}}
\put(27.5,29.5){\red{\line(1,0){35}}}
\put(28.2,33.5){\red{\line(1,0){33.6}}}
\put(27.6,28){\dgreen{\line(1,1){9.2}}}
\put(27.6,22){\dgreen{\line(1,1){15.2}}}
\put(28.7,16){\dgreen{\line(1,1){19.7}}}
\put(33,13.1){\dgreen{\line(1,1){18}}}
\put(40,13.1){\dgreen{\line(1,1){11.2}}}
\end{picture}
\caption{A schematic diagram for model classification of quantum parametric models. 
A generic quantum parametric model $\cM$ is indicated by the rectangular box. 
A blue vertically shadowed area represents the D-invariant model. 
A red horizontally shadowed area does the asymptotically classical model. 
A green diagonally shadowed area does the quasi classical model. 
The intersection of the D-invariant model and the asymptotically classical 
model represents the classical model. 
}
\label{fig1}
\end{figure}
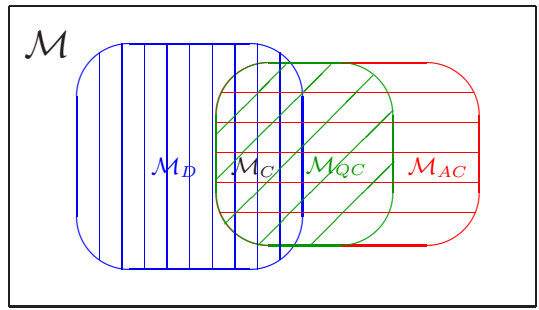

The content of this paper is summarized as follows. 
Sec.~\ref{sec2} provides preliminaries for notations and mathematical tools used in this paper. 
In Sec.~\ref{sec2-3}, a few lemmas are proven to be useful for classifying quantum statistical models. 
In Sec.~\ref{sec3}, we list the definitions of four different classes of statistical models. 
Our main results are given in the next section. Sec.~\ref{sec4-1} gives the main theorems of this paper. 
In Sec.~\ref{sec4-2}, we discuss the meaning of the classical model in detail. 
Proofs for the theorems are given in Sec.~\ref{sec4-3}. 
Several examples are discussed in Sec.~\ref{sec5} to illustrate our findings. 
The last section, Sec.~\ref{sec5}, concludes the paper with a few remarks and open problems. 

\section{Preliminaries}\label{sec2}
A {\it quantum system} $\cH$ is a $d$-dimensional Hilbert space on the complex number. 
Denote by $\lofh$ a set of (bounded) linear operators from $\cH$ to itself, 
and by $\lofhh$ a set of linear and hermite operators from $\cH$ to itself. 
A {\it quantum state} is a positive semi-definite operator on $\cH$ with unit trace. 
Let us denote a set of all quantum states on $\cH$ by $\barsofh$ and 
all full-ranked quantum states by $\sofh$. 
A {\it quantum statistical model} is defined by a parametric family of quantum states 
\be\label{qmodel}
\cM:=\{\rho_\theta\in\sofh\,|\,\theta=(\theta^1,\dots,\theta^n)\in\Theta\},
\ee
where $\Theta$ is an open subset of $\bbr^n$. 
As in classical statistics, we impose several regularity conditions, 
such as one-to-one smooth mapping; $\theta\mapsto\rho_\theta$, differentiability, 
linearly independence of partial derivatives $\del \rho_\theta/\del\theta^i$ with respect to 
the coordinates $(\theta^i)$, non-degeneracy for the eigenvalues, and so on. 
In the following discussions, we assume all these regularity conditions to avoid non-regular 
behaviors of the statistical model. 
In particular, we mainly consider a quantum statistical model of full-rank states unless stated explicitly. 

\subsection{Tangent space and quantum Fisher information}\label{sec2-1}
We define two quantum versions of the logarithmic derivative, the quantum score functions, as follows. 
For a given quantum state $\rho_{\theta}$ and any operators $X,Y\in\lofh$, 
define the {\it symmetric logarithmic derivative} (SLD) and 
{\it right logarithmic derivative} (RLD) inner product by 
\begin{align} \nonumber
\sldin{X}{Y}&:=\frac12\tr{\rho_{\theta}(YX^\dagger+X^\dagger Y)},\\
\rldin{X}{Y}&:=\tr{\rho_{\theta}YX^\dagger}, 
\end{align}
respectively, where $X^\dagger$ denotes the hermite conjugate of $X$. 
The $i$th SLD and RLD operators, 
$L_i$ and $\tilde{L}_i$, are formally defined by the solutions to the operator equations:
\begin{align}\nonumber
\del_i\rho_{\theta}&=\frac12 (\rho_{\theta}L_{\theta,i}+L_{\theta,i}\rho_{\theta}), \\ \label{sldrldop}
\del_i\rho_{\theta}&=\rho_{\theta}\tilde{L}_{\theta,i}. 
\end{align}
for $i=1,2,\dots,n$, where $\del_i:=\frac{\del}{\del\theta^i}$ denotes the partial derivative with respect to $\theta^i$. 
It is not difficult to see that the SLD operators are hermite, whereas 
RLD operators are not in general. 

The SLD and RLD Fisher information matrices (quantum Fisher metric) are defined by 
\begin{align} \nonumber
G_{\theta}&:= \left[ g_{\theta, ij}\right] \mbox{ with }
g_{\theta, ij}:=\sldin{L_{\theta,i}}{L_{\theta,j}}, \\
\tilde{G}_{\theta}&:= \left[ \tilde{g}_{\theta, ij} \right]\mbox{ with }
\tilde{g}_{\theta, ij}:=\rldin{\tilde{L}_{\theta,i}}{\tilde{L}_{\theta,j}}, 
\end{align}
respectively. 
It is known that the SLD Fisher information is the smallest and the real part of RLD is 
the largest operator monotone metrics on the quantum state space \cite{petz}. 

The SLD tangent space is define by the linear span of SLD operators:
\be
T_\theta(\cM):=\rspan{L_{\theta,i}}\subset\lofhh,
\ee
and the RLD tangent space is defined by the linear span of RLD operators 
with complex coefficients: 
\be
\tilde{T}_\theta(\cM):=\cspan{\tilde{L}_{\theta,i}}\subset\lofh. 
\ee
Let $G_\theta^{-1}=[g_\theta^{ij}]$ be the inverse of the SLD Fisher information 
and $\tilde{G}_\theta^{-1}=[\tilde{g}_\theta^{ij}]$ be the inverse for the RLD case. 
It is convenient to introduce the following linear combinations of 
the logarithmic derivative operators 
\[ 
L_{\theta}^i:= \sum_{j=1}^n g_\theta^{ji}{L}_{\theta,j},\ 
\tilde{L}_{\theta}^i:=\sum_{j=1}^n\tilde{g}_\theta^{ji}\tilde{L}_{\theta,j} . 
\]
By definitions, $\{L_{\theta}^i\}$ forms a dual basis for the inner product space $\sldin{\cdot}{\cdot}$;  
$\sldin{L_{\theta}^i}{L_{\theta,j}}=\delta^i_j$, and we shall call it the SLD dual operator.
The same statement holds for the RLD case. 

Noting that the SLD and RLD operators are a sort of exponential representation of 
the tangent vector $\del_i$, we can show the next lemma. 
\begin{lemma}\label{lem1} 
For $\forall X\in\lofh$, and $\forall f\in C^\infty(\bbr)$, the following holds. 
\be 
\sldin{f(L_{\theta,i})}{X}=\rldin{f(\tilde{L}_{\theta,i})}{X}, \label{prop2}\\
\ee
\end{lemma}
\begin{proof}
We note that the definitions of logarithmic derivative operators gives 
\be
\sldin{L_{\theta,i}}{X}=\rldin{\tilde{L}_{\theta,i}}{X}=\tr{\del_i\rho_\theta X}, 
\ee
and repeated applications of this relation proves
\be
\sldin{(L_{\theta,i})^k}{X}=\rldin{(\tilde{L}_{\theta,i})^k}{X}, 
\ee
for any integer power $k$. It is then easy to prove Eq.~\eqref{prop2}. 
$\square$
\end{proof}

\subsection{Commutation operator}\label{sec2-2}
For a given quantum statistical model \eqref{qmodel}, we introduce 
a super-operator $\cD$ from $\lofh$ to itself, whose action on $X\in\lofh$ 
is defined by the operator equation: 
\be\label{defDop}
[\rho_\theta\,,\,X]:=\rho_\theta X-X\rho_\theta=\I\rho_\theta\Dop{X}+\I\Dop{X}\rho_\theta.  
\ee
The operator $\cD_{\rho_\theta}$, called a {\it commutation operator}, was introduced by Holevo, 
and the detail can be found in his book \cite{holevo}. 
By definition, we can check that the operator $\cD_{\rho_\theta}$ is linear. 
Denoting the identity operator $I$, the following relationship holds 
\be
L_{\theta,i}=(I+\I\cD_{\rho_\theta})(\tilde{L}_{\theta,i}), \label{prop1}
\ee
which can be proven by the direct calculation. 

The properties useful in our discussion are given in the next lemma. 
\begin{lemma}\label{lem2} 
For $\forall X,Y\in\lofh$, the following relations hold. 
\begin{align} 
\sldin{\Dop{X}}{Y}&=-\sldin{{X}}{\Dop{Y}}, \label{prop3}\\
\rldin{\Dop{X}}{Y}&=-\rldin{{X}}{\Dop{Y}} . \label{prop4}
\end{align} 
\end{lemma}
\begin{proof}
The first relationship can be proven directly as
\begin{align*}
2\sldin{\Dop{X}}{Y}&=\tr{\rho_\theta (\Dop{X}Y+Y\Dop{X})}\\
&=\tr{(\rho_\theta\Dop{X}+\Dop{X}\rho_\theta)Y}\\
&= \tr{(-\I)[\rho_\theta,X]Y}\\
&=-\tr{(-\I)[\rho_\theta,Y]X}\\
&=-\tr{(\rho_\theta\Dop{Y}+\Dop{Y}\rho_\theta)X}\\
&=-2\sldin{{X}}{\Dop{Y}}.
\end{align*}

Eq.~\eqref{prop4} can be proven similarly. 
$\square$
\end{proof}

\subsection{Basic lemmas} \label{sec2-3}
In this subsection, we list several lemmas that will be used in our discussion. 
We define two hermite matrices, ${Z}_\theta,\tilde{Z}_\theta$ in terms of SLD and RLD dual operators as follows. 
\begin{align} 
Z_\theta&:=[z_\theta^{ij}],\ \mbox{with } z_\theta^{ij}:=\rldin{L_\theta^i}{L_\theta^j},\\
\tilde{Z}_\theta&:=[\tilde{z}_\theta^{ij}],\ \mbox{with }\tilde{z}_\theta^{ij}:=\sldin{\tilde{L}_\theta^i}{\tilde{L}_\theta^j}. 
\end{align}
By definition, they are complex matrices in general. Hermiteness can be 
checked directly by
\be
(z_\theta^{ij})^*=\tr{(\rho_\theta L_\theta^j {L_\theta^i}^\dagger)^\dagger }=\tr{\rho_\theta L_\theta^i {L_\theta^j} }=z_\theta^{ji}, 
\ee
where $*$ denotes its complex conjugate, and the matrix $\tilde{Z}_\theta$ can be checked similarly. 

Together with the SLD and RLD Fisher information matrices, 
we list four matrices for comparison: 
\begin{align} \label{4matrix}
G_\theta^{-1}&=[g_\theta^{ij}],\ g_\theta^{ij}=\sldin{L_{\theta}^i}{L_{\theta}^j},\\ \nonumber
\tilde{G}_\theta^{-1}&=[\tilde{g}_\theta^{ij}],\ \tilde{g}_\theta^{ij}=\rldin{\tilde{L}_{\theta}^i}{\tilde{L}_{\theta}^j},\\ \nonumber
Z_\theta&=[z_\theta^{ij}],\ z_\theta^{ij}=\rldin{L_\theta^i}{L_\theta^j},\\ \nonumber
\tilde{Z}_\theta&=[\tilde{z}_\theta^{ij}],\ \tilde{z}_\theta^{ij}=\sldin{\tilde{L}_\theta^i}{\tilde{L}_\theta^j}. 
\end{align}
By definition, $\Re (Z_\theta^{-1})=G_\theta$ and $\Re Z_\theta=G_\theta^{-1}$ hold, 
where $\Re X:=(X+X^*)/2$ denotes the real part of $X\in\lofh$ with $X^*$ the complex conjugate of $X$. 

First, it is straightforward to see that the operator $L_\theta^i-\tilde{L}_\theta^i$ has the following property. 
\begin{lemma}\label{lem3}
$L_\theta^i-\tilde{L}_\theta^i$ is orthogonal to the SLD tangent space $T_\theta(\cM)$ 
with respect to $\sldin{\cdot}{\cdot}$, 
and is orthogonal to the RLD tangent space $\tilde{T}_\theta(\cM)$ 
with respect to $\rldin{\cdot}{\cdot}$. 
\end{lemma}
\begin{proof}
Direct calculation shows
\begin{align*}
\sldin{L_{\theta,j}}{L_\theta^i-\tilde{L}_\theta^i}&=\sldin{L_{\theta,j}}{L_\theta^i}-\sldin{L_{\theta,j}}{\tilde{L}_\theta^i}\\
&=\sldin{L_{\theta,j}}{L_\theta^i}-\rldin{\tilde{L}_{\theta,j}}{\tilde{L}_\theta^i}\\
&=\delta^i_j-\delta^i_j=0,
\end{align*}
where Lemma \ref{lem1} with $f(x)=x$ is used to get the second line. 

Orthogonality to the RLD tangent space with respect to the RLD inner product can be proven similarly. 
$\square$
\end{proof}

The following matrix inequalities between $G_\theta$, $\tilde{G}_\theta$, $Z_\theta=[\rldin{L_\theta^i}{L_\theta^j}]$, 
and $\tilde{Z}_\theta=[\sldin{\tilde{L}_\theta^i}{\tilde{L}_\theta^j}]$ are fundamental. 
\begin{lemma}\label{lem4} Two matrix inequalities 
\begin{align}\nonumber
Z_\theta&\ge \tilde{G}_\theta^{-1},\\
\tilde{Z}_\theta&\ge G_\theta^{-1}, 
\end{align}
hold where the equality conditions are same and is given by $\forall i$, $L_\theta^i-\tilde{L}_\theta^i=0$. 
\end{lemma}
\begin{proof}
Let $m_{\theta}^i:=L_\theta^i-\tilde{L}_\theta^i$ and define an $n\times n$ hermite matrix, 
\be
\tilde{M}_\theta:= [\rldin{m_{\theta}^i}{m_{\theta}^j}].
\ee
The matrix $\tilde{M}_\theta$ is then positive semi-definite. 
Using Lemma \ref{lem3}, we can also express matrix elements of $\tilde{M}_\theta$ as
\begin{align*}
\rldin{m_{\theta}^i}{m_{\theta}^j}&=\rldin{m_{\theta}^i}{L_\theta^j}-\rldin{m_{\theta}^i}{\tilde{L}_\theta^j}\\
&=\rldin{L_\theta^i}{L_\theta^j}-\rldin{\tilde{L}_\theta^i}{L_\theta^j}\\
&=z_\theta^{ij} -\sum_k \tilde{g}_\theta^{ik}  \rldin{\tilde{L}_{\theta,k}}{L_\theta^j}\\
&=z_\theta^{ij} -\sum_k \tilde{g}_\theta^{ik}  \sldin{L_{\theta,k}}{L_\theta^j}\\
&=z_\theta^{ij}-\tilde{g}_\theta^{ij}, 
\end{align*}
where second equality is due to Lemma \ref{lem3}. 
third equality follows from definition of the RLD dual operator. 
Fourth equality is due to Lemma \ref{lem1}. 
Therefore, we show the matrix inequality $\tilde{M}_\theta=Z_\theta- \tilde{G}_\theta^{-1}\ge0$. 
The equality is satisfied if and only if this matrix $\tilde{M}_\theta$ is zero. 
This is equivalent to $m_\theta^i=L_\theta^i-\tilde{L}_\theta^i=0$ for all $i=1,2,\dots,n$. 

The second inequality can be proven in the same way by starting with 
$M_\theta:= [\sldin{m_{\theta}^i}{m_{\theta}^j}]$. 
$\square$
\end{proof}

Next, define $m_{\theta,i}:=L_{\theta,i}-\tilde{L}_{\theta,i}$ and consider another 
hermite matrix $M_\theta:=[\sldin{m_{\theta,i}}{m_{\theta,j}}]$. 
Following exactly the same logic as in Lemma \ref{lem4},  we can prove the next lemma.  
\begin{lemma}\label{lem4-2} Two matrix inequalities 
\begin{align}\nonumber
G_\theta+\tilde{G}_\theta\tilde{Z}_\theta\tilde{G}_\theta&\ge 2\tilde{G}_\theta,\\
\tilde{G}_\theta +G_\theta Z_\theta G_\theta&\ge 2G_\theta, 
\end{align}
hold where the equality conditions are same and is given by $\forall i$, $L_{\theta,i}-\tilde{L}_{\theta,i}=0$. 
\end{lemma}

Finally, the commutation operator and logarithmic operators satisfy the following relations \cite{comment}.  
Importantly, the right hand side of three equations are expressed as a difference 
between two hermite matrices defined in Eqs.~\eqref{4matrix}. 
\begin{lemma}\label{lem5} 
\begin{align} 
\sldin{L_{\theta}^i}{\I\Dop{L_{\theta}^j}}&=z_\theta^{ij}-g_\theta^{ij}=\I\,\Im z_\theta^{ij},\label{prop5}\\ 
\sldin{\tilde{L}_{\theta}^i}{\I\Dop{L_{\theta}^j}}&=\tilde{g}_\theta^{ij}-g_\theta^{ij},\label{prop6}\\
\sldin{\tilde{L}_{\theta}^i}{\I\Dop{\tilde{L}_{\theta}^j}}&=\tilde{g}_\theta^{ij}-\tilde{z}_\theta^{ij},\label{prop7}
\end{align}
hold for $\forall i,j$. 
\end{lemma}
\begin{proof}
Using definitions of the SLD and RLD inner product, and the commutation operator, we have
\begin{align*}
\rldin{X}{Y}-\sldin{X}{Y}&=\frac12 \tr{\rho_\theta[Y\,,\,X^\dagger]}\\
&=\frac12 \tr{[\rho_\theta\,,\,Y]X^\dagger}\\
&=\frac\I2 \tr{(\rho_\theta\Dop{Y}+\Dop{Y}\rho_\theta)X^\dagger}\\
&=\frac\I2 \tr{\rho_\theta(\Dop{Y}X^\dagger+X^\dagger\Dop{Y}}\\
&=\frac12 \sldin{X}{\I\Dop{Y}}, 
\end{align*}
for all $X,Y\in\lofh$. 
Setting $X=L_\theta^i, Y=L_\theta^j$, we prove Eq.~\eqref{prop5}. 
Similarly, $X=\tilde{L}_\theta^i, Y=\tilde{L}_\theta^j$ gives Eq.~\eqref{prop7}. 

Next, we observe 
\be
g^{ij}_\theta=\sldin{\tilde{L}_\theta^i}{L_\theta^j},\ \tilde{g}_\theta^{ij}=\rldin{\tilde{L}_\theta^i}{L_\theta^j}, 
\ee
which can be directly checked. These relations then prove Eq.~\eqref{prop6}. 
$\square$
\end{proof}

\section{Model class in quantum parametric models}\label{sec3}
In this section, we consider four different classes for quantum parametric models. 
The first class is a purely classical. The second class is so called a commutative model. 
The third and fourth ones are nontrivial, the D-invariant and asymptotically classical models. 
\subsection{Classical model}
At each point $\theta\in\Theta$, the quantum state $\rho_\theta$ can be diagonalized with 
a unitary $U_\theta$ as $\rho_\theta=U_\theta \Lambda_\theta U_\theta^{-1}$, 
where a diagonal matrix, 
\be
\Lambda_\theta=\left(\begin{array}{cccc}p_\theta(1) & 0 & \cdots & 0 \\0 & p_\theta(2) & \cdots & 0 \\
\vdots & \vdots & \ddots & \vdots \\0 & 0 & \cdots & p_\theta(d)\end{array}\right)
\ee 
lists the eigenvalues of the state $\rho_\theta$. 
By definition, $\forall i,\,p_\theta(i)>0$ and $\sum_{i=1}^dp_\theta(i)=1$. 
In other words, $\Lambda_\theta$ can be regarded as an element of 
$\cP(d):=$ the set of all (positive) probability distributions on the set $\{1,2,\dots,d\}$. 
When the unitary $U_\theta$ is independent of $\theta$ for all point in $\Theta$, 
it is clear that any statistical problem is reduced to the classical one. With this identification, 
we have the following definition. 
\begin{definition}[Classical statistical model]\label{def_cmodel}
For a given parametric quantum statistical model \eqref{qmodel}, 
the model is said {\it classical} if the family of quantum states $\rho_\theta$ 
can be diagonalized with a $\theta$-independent unitary $U$ as
\be\label{cmodel_decomp}
\rho_\theta=U \Lambda_\theta U^{-1},
\ee
for all parameter values $\theta\in\Theta$. 
\end{definition}
%
In the following, we denote the set of all classical models on $\cH$ by $\cM_{C}$. 

\subsection{Quasi classical model}
The second class of quantum statistical models has been known in the literature. 
It is called a {\it quasi classical} or {\it commutative} model. 
\begin{definition}[Quasi classical model]\label{def_cmodel}
A parametric quantum statistical model \eqref{qmodel} is said {\it quasi classical}, 
if all SLD operators commute with each other at all point $\theta$. That is, 
\be
[L_{\theta,i}\,,\,L_{\theta,j}]=0,\ \forall i, j, 
\ee
hold for all parameter values $\theta\in\Theta$. 
\end{definition}

Clearly, if the model is classical, then it is also quasi classical. 
However, the converse statement does not hold in general. 
A simple counter example is discussed in Sec.~\ref{sec-non_cl}. 
It is also easy to see that any one-parameter model is automatically quasi-classical. 

An important property of quasi classical models is that 
we can diagonalize all SLD operators simultaneously. 
It is then possible to perform a measurement that saturates the SLD CR bound 
defined in Eq.~\eqref{sldbound} explicitly. 
Let us denote the set of all quasi classical models on $\cH$ by $\cM_{QC}$. 

\subsection{Asymptotic bound: Holevo bound}
In this subsection, we give a brief summary of the asymptotic theory on quantum state estimation \cite{hayashiWS}. 
As in classical statistics, we are given $N$-tensor product of identically and independently distributed (i.i.d.) 
quantum states $\rho_\theta^{\otimes N}:=\rho_\theta\otimes\rho_\theta\otimes\cdots\otimes\rho_\theta$ on $\cH$. 
We perform a measurement $\hat{\Pi}^{(N)}$ on $\rho_\theta^{\otimes N}$, 
which is described by a set of matrices under certain conditions, to infer an unknown parameter value $\theta$. 
The estimation error of the measurement $\hat{\Pi}^{(N)}$ is evaluated by the standard mean-square error (MSE) matrix 
$V^{(N)}_\theta[\hat\Pi^{(N)}]$. In the asymptotic theory of quantum state estimation, one minimizes 
the weighted trace of the MSE matrix under an additional condition as follows.  
\be \label{CRbound}
C_\theta[W]:=\inf_{ \{\hat{\Pi}^{(N)}\}\mbox{ is a.u.}}
\big\{ \limsup_{N\to\infty}\,N\,\Tr{W V^{(N)}_\theta[\hat{\Pi}^{(N)}]}\big\}, 
\ee
where $W>0$ is an arbitrary positive-definite weight matrix and a.u. stands for asymptotically unbiased.   
The first order estimation error bound \eqref{CRbound} is usually referred to as the Cram\'er-Rao (CR) type bound in the literature. 
There have been many mathematical works to obtained an alternative expression for the CR bound 
in terns of information quantities, such as the quantum Fisher information matrix. 
Unlike classical statistics, where the bound is given by the Fisher information matrix, 
the above bound cannot be written as a simple closed formula in general. 
However, it takes the following optimization form known as the Holevo bound: 
\be \label{hbound}
C_\theta^H[W]:=\inf_{\vec{X}\in\cX_\theta}h_\theta[\vec{X}|W]. 
\ee
In this definition, the set $\cX_\theta$ is defined by 
\begin{multline}\nonumber
\cX_\theta:=\{\vec{X}=(X^1,X^2,\dots, X^n)\,|\,\forall i\,X^i\in\lofhh,\\
 \forall i\,\tr{\rho_\theta X^i}=0, \forall i,j\,\tr{\frac{\del \rho_\theta}{\del\theta^i} X^j}=\delta^j_{\,i} \}.
\end{multline}
Introduce an $n\times n$ hermite matrix $H_\theta[\vec{X}]:= \big[ \rldin{X^i}{X^j} \big]$,  
and we define the function $h_\theta[\vec{X}|W]$ by
\be \nonumber
h_\theta[\vec{X}|W]:=\Tr{W \Re H_\theta[\vec{X}]}+\Tr{|W^{\frac12} \Im H_\theta[\vec{X}] W^{\frac12}|}, 
\ee
where $|X|=\sqrt{X^\dagger X}$ denotes the absolute value of a linear operator $X$, 
and $\Im X:=(X-X^*)/2\I$ denotes the imaginary part of $X\in\lofh$.   
The following theorem establishes that the Holevo bound is equal to the CR type bound. 
\begin{theorem}\label{thm1}
For a quantum statistical model satisfying the regularity conditions, 
$C_\theta[W]=C_\theta^H[W]$ holds for all weight matrices. 
\end{theorem}

Proofs based on different assumptions can be found in Refs.~\cite{\QlanRef}. 
The Holevo bound is regarded as unification of previously known bounds \cite{nagaoka89}, 
such as the SLD and RLD CR bounds: 
\begin{align}\label{sldbound}
C_\theta^S[W]&:=\Tr{WG_\theta^{-1}},\\  \label{rldbound}
C_\theta^R[W]&:=\Tr{W\Re\tilde{G}_\theta^{-1}}+\Tr{|W^{\frac12} {\Im}\tilde{G}_\theta^{-1} W^{\frac12} |}. 
\end{align} 
The relation ship $C_\theta^H[W]\ge\max\{ C_\theta^S[W],C_\theta^R[W]\}$ holds for all $W>0$ \cite{holevo}.  

\subsection{D-invariant model}
Holevo introduced an important class of quantum statistical models based on 
the commutation operator $\cD_{\rho_\theta}$ \cite{holevo}.  
\begin{definition}[D-invariant model (Holevo)]\label{def_Dinv1}
A quantum statistical model \eqref{qmodel} is called {\it D-invariant} at $\theta$, 
if the SLD tangent space at $\theta$ is an invariant subspace of the commutation operator. 
\end{definition}

Mathematically, this condition is expressed as $\forall X\in T_\theta(\cM)$, $\Dop{X}\in T_\theta(\cM)$ at $\theta$. 
For our discussion, we will focus on the D-invariant model at all $\theta$ (global D-invariance). 
For (globally) D-invariant models, the Holevo bound can be expressed analytically and coincides with 
the RLD CR bound \cite{holevo}, i.e., $\forall W>0$, $C^H_\theta[W]=C^R_\theta[W]$, 
and its achievability was discussed in the literature. 

Based on the result of Ref.~\cite{jsJMP}, we have another definition 
for the D-invariant model. 
\begin{definition}[D-invariant model 2]
A quantum statistical model \eqref{qmodel} is called {\it D-invariant} at $\theta$, 
if the Holevo bound is identical to the RLD CR bound for all positive weight matrices.
\end{definition}
The equivalence between two definitions was proven \cite{jsJMP}. 
\begin{theorem}\label{thm2}
The Holevo bound is identical to the RLD CR bound for all weight matrices, 
if and only if the quantum statistical model is D-invariant in the sense of Definition \ref{def_Dinv1}. 
\end{theorem} 
The set of all D-invariant models is denoted by $\cM_D$. 

\subsection{Asymptotically classical model}
The last class of quantum statistical models is when the Holevo bound coincides with the SLD CR bound. 
\begin{definition}
A quantum statistical model \eqref{qmodel} is called {\it asymptotically classical}, 
if the Holevo bound is identical to the SLD CR bound for all positive weight matrices. 
\end{definition}
Mathematically, this definition is expressed by the condition: $\forall W>0$, $C_\theta^H[W]=C_\theta^S[W]$. 
We shall denote the set of all asymptotically classical models by $\cM_{AC}$.

\section{Model classification and characterization}\label{sec4}
In this section, we give classification of quantum statistical models 
based on the notations and concept introduced in Sec.~\ref{sec3}. 
We first list the results on several equivalent characterization of each model class. 
Discussions on the results are presented followed by the proofs in Sec.~\ref{sec4-3}. 

\subsection{Results}\label{sec4-1}
\subsubsection{Classical model}
The following theorem characterizes the classical model.
\begin{theorem}\label{thm4}
For a given (regular) quantum statistical model \eqref{qmodel}, the following conditions are all equivalent. 
\begin{enumerate}
\item The model is classical (Def.~\ref{def_cmodel}). \label{eq1}
\item $\forall X\in T_\theta(\cM)$, $[X\,,\,\rho_\theta]=0$. \label{eq2}
\item $\forall X\in \tilde{T}_\theta(\cM)$, $[X\,,\,\rho_\theta]=0$. \label{eq3}
\item $G_\theta=\tilde{G}_\theta$. \label{eq4}
\item $\forall i$, ${L}_{\theta,i}=\tilde{L}_{\theta,i}$. \label{eq5}
\item $\Dop{T_\theta(\cM)}=0$. \label{eq6}
\item $\Dop{\tilde{T}_\theta(\cM)}=0$. \label{eq7}
\item The model is D-invariant and asymptotically classical. \label{eq8}
\end{enumerate}\label{thm5}
Here we remind that all statements are made for global aspect of the model, 
that is for all point $\theta\in\Theta$. 
\end{theorem}

\subsubsection{D-invariant model}
In Ref.~\cite{jsJMP}, we derived several equivalent characterizations of the D-invariant model, 
which are summarized in the following theorem. 
\begin{theorem}\label{lem6}
Given a quantum statistical model, the following conditions are equivalent. 
\begin{enumerate}
\item $\cM$ is D-invariant at $\theta$.
\item $\forall i$, $\cD_{\rho_\theta}(L_\theta^i)=\sum_j (\Im Z_\theta)^{ji}L_{\theta,j}$. 
\item $Z_\theta=\tilde{G}_\theta^{-1}$ 
\item $\forall i$, $L_\theta^i=\tilde{L}_\theta^i$. 
\item $\forall X^i\in\lofhh$, $X^i-L_\theta^i\bot T_\theta(\cM)$ with respect to $\sldin{\cdot}{\cdot}$ 
$\Rightarrow$ $X^i-L_\theta^i\bot T_\theta(\cM)$ with respect to $\rldin{\cdot}{\cdot}$. 
\end{enumerate}
\end{theorem}

\subsubsection{Asymptotically classical model}
With this notion of the asymptotically classical model, we have the following result. 
\begin{theorem}\label{thm3}
For a regular quantum statistical model, the following equivalences hold:
\begin{enumerate}
\item $\cM$ is asymptotically classical. \label{ac1}
\item $\exists W_0>0$, $C_\theta^H[W_0]=C_\theta^S[W_0]$. \label{ac2}
\item $Z_\theta=G_\theta^{-1}\ (\lra\Im Z_\theta=0)$. \label{ac3}
\item $\forall i,j$, $\tr{\rho_\theta[L_{\theta,i},L_{\theta,j}]}=0$. \label{ac4}
\end{enumerate}
\end{theorem}
We note that the equivalence among conditions \ref{ac1}, \ref{ac3}, and \ref{ac4} were 
presented in Ref.~\cite{jsqit32}. 
Equivalence between the first and the last conditions (\ref{ac1} and \ref{ac4}) was 
independently proven in Ref.~\cite{RJDD16}, in which the authors named the 
``compatibility condition.'' 
The last condition \ref{ac2} was suggested by Nagaoka \cite{nagaoka_com}. 

\subsubsection{$G_\theta^{-1},\tilde{G}_\theta^{-1},Z_\theta,\tilde{Z}_\theta$ matrices}
Combining the previous lemmas and theorems with additional analysis, 
we can obtain another interesting characterizations of quantum statistical models 
based on the four hermite matrices, $G_\theta^{-1},\tilde{G}_\theta^{-1},Z_\theta,\tilde{Z}_\theta$. 
This is summarized in the next corollary. 
\begin{corollary} \label{cor1}
Given a quantum statistical model, we have the following equivalences.
\begin{enumerate}
\item $\cM$ is classical. \\ \label{cond_cl}
$\lra$ $G_\theta^{-1}=\tilde{G}_\theta^{-1}$ $\lra$  $\tilde{G}_\theta^{-1}=\tilde{Z}_\theta$ $\lra$  $Z_\theta=\tilde{Z}_\theta$
\item $\cM$ is D-invariant. \\ \label{cond_Dinv}
$\lra$ $\tilde{G}_\theta^{-1}=Z_\theta$ $\lra$  ${G}_\theta^{-1}=\tilde{Z}_\theta$
\item $\cM$ is asymptotically classical. \\ \label{cond_ac}
$\lra$ $G_\theta^{-1}=Z_\theta$ 
\end{enumerate}
\end{corollary}
Figure \ref{fig2} gives a schematic diagram summarizing the relations among the matrices $G_\theta^{-1},\tilde{G}_\theta^{-1},Z_\theta,\tilde{Z}_\theta$. 
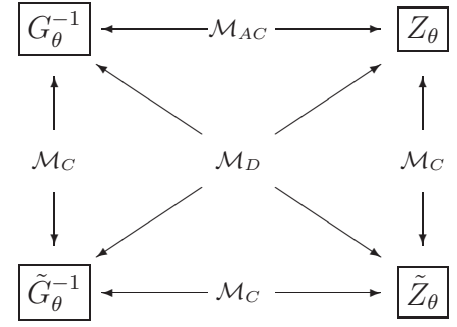
\begin{figure}
\setlength{\unitlength}{0.7mm}
\begin{picture}(80,60)(-10,-5)
\linethickness{0.35pt}
\put(29,0){\vector(-1,0){20}}
\put(42,0){\vector(1,0){20}}
\put(29,50){\vector(-1,0){20}}
\put(42,50){\vector(1,0){20}}
\put(0,19){\vector(0,-1){10}}
\put(0,31){\vector(0,1){10}}
\put(70,19){\vector(0,-1){10}}
\put(70,31){\vector(0,1){10}}
\put(42,30){\vector(3,2){20}}
\put(42,20){\vector(3,-2){20}}
\put(28,30){\vector(-3,2){20}}
\put(28,20){\vector(-3,-2){20}}
\put(0,0){\makebox(0,0)[c]{\large\fbox{$\tilde{G}_\theta^{-1}$}}}
\put(0,50){\makebox(0,0)[c]{\large\fbox{$G_\theta^{-1}$}}}
\put(70,50){\makebox(0,0)[c]{\large\fbox{$Z_\theta$}}}
\put(70,0){\makebox(0,0)[c]{\large\fbox{$\tilde{Z}_\theta$}}}
\put(35,25){\makebox(0,0)[c]{$\cM_{D}$}}
\put(0,25){\makebox(0,0)[c]{$\cM_{C}$}}
\put(70,25){\makebox(0,0)[c]{$\cM_{C}$}}
\put(35,0){\makebox(0,0)[c]{$\cM_{C}$}}
\put(35,50){\makebox(0,0)[c]{$\cM_{AC}$}}
\end{picture}
\caption{A schematic diagram for model classification for three classes: the classical ($\cM_{C}$), D-invarinat ($\cM_{D}$),  and 
asymptotically classical ($\cM_{AC}$) in terms of four matrices $G_\theta^{-1},\tilde{G}_\theta^{-1},Z_\theta,\tilde{Z}_\theta$. 
Two arrows in the opposite direction indicate if two matrices are identical, 
then a model belongs to a class indicated between these arrows.}
\label{fig2}
\end{figure}

\subsection{Discussion on Theorem \ref{thm4}}\label{sec4-2}
In this subsection, we discuss the meaning and its statistical consequences of Theorem \ref{thm4}. 
\subsubsection{Tangent vector}
We first note that two conditions \ref{eq2} and \ref{eq3} are 
nothing but condition \eqref{cond2}. This is straightforward 
to understand if we regard $\del \rho_\theta/\del\theta^i$ as 
an m-representation of the tangent vector $\del/\del\theta^i$ 
and $L_{\theta,i}$ as an e-representation of it with respect to the SLD Fisher metric. 
The statement applies for the RLD case. 

\subsubsection{Quantum Fisher information}
Condition \ref{eq4} states that two quantum Fisher information matrices are identical. 
If this is the case, in fact, all possible monotone metric on $\sofh$ are identical. 
In other words, they collapse to the single monotone metric. 
This is due to the facts that 1) the imaginary part of the RLD Fisher information 
vanishes, and 2) the SLD Fisher is the minimum and the real RLD Fisher is the maximum 
monotone metric (Petz's theorem) \cite{petz}. 

We note that this result, equivalence between condition 1 and condition 4, was also 
stated in Ref.~\cite{matsumoto_thesis}. 

Next, we can contrast condition \ref{eq5} to the condition for a D-invariant model in Lemma \ref{lem1}: 
$L_\theta^i=\tilde{L}_\theta^i$ for all $i$. This latter condition is not equivalent to 
$G_\theta=\tilde{G}_\theta$ in general unless the imaginary part of the RLD Fisher information vanishes. 
Thus, condition \ref{eq5} is a stronger condition than the condition for the D-invariant model as should be. 

\subsubsection{Tangent space}
Condition \ref{eq6} (or \ref{eq7}) means that the SLD tangent space is in 
the kernel of the commutation operator $\cD$. 
We split the SLD operator into two parts; a classical part and quantum part 
where the latter is defined by the change in a unitary direction.  
Since the $\cD$ operator maps the commutation relationship 
to the anti-commutation relationship as in Eq.~\eqref{defDop}, 
the quantum part of the SLD operator is expressed in 
terms of the commutation operator. With more analysis, 
we can show that the condition for the classical model 
is equivalent to vanishing of the quantum part of SLD operators. 
See also discussion given in Ch.~7 of the book \cite{ANbook}. 

\subsubsection{Asymptotic bound}
The last equivalent condition is a rather straightforward consequence 
once we combining all ingredients presented in the lemmas and other equivalent conditions 
for the classical model. However, the statistical implication of this condition is 
non-trivial in the sense that we only consider properties of 
asymptotically achievable bounds. One is the bound for the D-invariant model, 
and the other is the bound for the asymptotically classical model. 
Another implication of this equivalence is that there is no 
genuine quantum statistical model that is both D-invariant and asymptotically classical. 

\subsection{Proofs}\label{sec4-3}
\subsubsection{Proof for Theorem \ref{thm4}}
We give a proof for Theorem \ref{thm4}. 
As we stated before, all conditions below are about all parameter values $\theta$ 
unless otherwise stated. 

\noindent
{\it Equivalence to \ref{eq2} and \ref{eq3}:}\\
First, we note that the definition of the classical model is equivalent to 
the commutativity of $\rho_\theta$ for all different values $\theta$, that is,  
\be\label{cond1}
[\rho_\theta\,,\,\rho_{\theta'}]=0\mbox{ for all }\theta\neq\theta'.
\ee
By the standard matrix analysis, this is equivalent to: 
\be\label{cond2}
\forall i,\ [\frac{\del}{\del\theta^i}\rho_{\theta}\,,\,\rho_\theta]=0.  
\ee
From the definitions of the SLD and RLD operators, 
we can show that condition \eqref{cond2} is equivalent to 
$[L_{\theta,i}\,,\,\rho_\theta]=0$ for all $i$. This is condition \ref{eq2}. 
Similarly, condition \eqref{cond2} can be converted to 
$[\tilde{L}_{\theta,i}\,,\,\rho_\theta]=0$ for all $i$, 
which is condition \ref{eq3}. $\square$

\noindent
{\it Equivalence to \ref{eq4} and \ref{eq5}:}\\
If the model is classical, the SLD operator $L_{\theta,i}$ commutes with 
the quantum state. Hence, operator equations \eqref{sldrldop} defining 
the SLD and RLD operators are identical. Since the SLD and RLD operator 
are uniquely defined, we obtain ${L}_{\theta,i}=\tilde{L}_{\theta,i}$ for all $i$. 

Next, assume condition \ref{eq5}, then matrices $\tilde{G}_\theta$ and $Z_\theta^{-1}$ 
are identical. Noting $\Re Z_\theta^{-1}=G_\theta$, we get condition \ref{eq4}. 

Last, suppose condition \ref{eq4}, $G_\theta=\tilde{G}_\theta$, then 
from Lemma \ref{lem3}, this is possible if and only if 
$\Im Z_\theta=0$ and $L_\theta^i=\tilde{L}_\theta^i$ for all $i$. 
Since $g_{\theta,ij}=\tilde{g}_{\theta,ij}$, the latter condition 
leads to ${L}_{\theta,i}=\tilde{L}_{\theta,i}$ for all $i,j$. $\square$

\noindent
{\it Equivalence to \ref{eq6} and \ref{eq7}:}\\
Condition \ref{eq6} is to say that the SLD tangent space 
is in the kernel of the commutation operator. 
From definition of the commutation operator 
and the fact that $X\rho+\rho X=0$ implies $X=0$ if $\rho>0$,  
we have
\[
\ker \cD_{\rho_\theta}=\{X\in\lofh\,|\, [X,\rho_\theta]=0 \}. 
\]
This then immediately establishes equivalence between condition \ref{eq2} and condition \ref{eq6}. 
A similar argument applies for condition \ref{eq7}. $\square$ 

\noindent
{\it Equivalence to \ref{eq8}}: \\
When the model is classical, conditions \ref{eq4} and \ref{eq5} 
give $L_\theta^i=\tilde{L}_\theta^i$ for all $i$ (D-invariance). 
Combining it with ${L}_{\theta,i}=\tilde{L}_{\theta,i}$ leads to 
$Z_\theta=G_\theta^{-1}$. 
Hence, the definitions for D-invariant and asymptotically classical model 
are clearly satisfied, if the model is classical. 
Conversely, suppose that the model is D-invariant, $\tilde{G}_\theta^{-1}=Z_\theta$, 
and asymptotically classical, $Z_\theta=G_\theta^{-1}$. 
Then, it gives condition \ref{eq4}, $ G_\theta=\tilde{G}_\theta$. $\square$

\subsubsection{Proof for Theorem \ref{thm3}}
\begin{proof} 
First: The third condition $\Im Z_\theta=0$ implies $\forall W>0,\, C_\theta^H[W]= C_\theta^S[W]$. 
This is because of $C_\theta^H[W]\ge C_\theta^S[W], \forall W>0$ 
and the direct substitution gives \\
$h_\theta[\vec{L}_{\theta}|W]= C_\theta^S[W]+\Tr{|W^{\frac12}\Im Z_\theta W^{\frac12}|}=C_\theta^S[W]$. \\
Here $\vec{L}_{\theta}=(L^1_\theta,L^2_\theta,\dots, L^n_\theta)\in\cX_\theta$ 
is the collection of SLD dual operators. This means the set of SLD dual operators is 
the optimal achieving the lowest value in the definition of the Holevo bound \eqref{hbound}.  

By definition, the first condition obviously implies the second one: 
$\exists W_0>0$, $C_\theta^H[W_0]= C_\theta^S[W_0]$. 

To show that the existence of a weight matrix $W_0$ satisfying $C_\theta^H[W_0]= C_\theta^S[W_0]$ 
implies the vanishing of the imaginary part of the matrix $Z_\theta$, we prove the contraposition. 
That is, if $ \Im Z_\theta\neq\v0$, then $C_\theta^H[W]> C_\theta^S[W]$ holds for all weight matrices $W$. 
Let us use the following substitution for optimizing the Holevo function: 
\be
\vec{X}=(L^1_\theta,L^2_\theta,\dots, L^n_\theta)+(K^1_\theta,K^2_\theta,\dots, K^n_\theta), 
\ee
where $K^i_\theta$ ($i=1,2,\dots,n$) are tangent operators orthogonal 
to all SLD operators $L_{\theta,i}$ with respect to the SLD inner product. 
With this, the Holevo function reads
\begin{multline}
h_\theta[\vec{X}|W]=C_\theta^S[W]+\Tr{W\Re \cK_\theta}\\
+\Tr{|W^{\frac12}\Im(Z_\theta+\cK_\theta)W^{\frac12} |}, 
\end{multline}
where $n\times n$ matrix $\cK_\theta=\big[\rldin{K^i_\theta}{K^j_\theta}\big]$ is hermite. 
We note that the last two terms: \\
$\Tr{W\Re \cK_\theta}+\Tr{|W^{\frac12}\Im(Z_\theta+\cK_\theta)W^{\frac12}|}$ is strictly positive 
since it vanishes if and only if $\Re \cK_\theta=0$ and $\Im(Z_\theta+\cK_\theta)=0$ hold. 
But these two conditions cannot be satisfied due to the assumption $Z_\theta\neq  0$ 
and the positivity of the matrix $\cK_\theta$. Therefore, we show that if $\Im Z_\theta\neq0$, 
we have $C_\theta^H[W]> C_\theta^S[W]$ for all $W>0$. 
Finally, $\Im Z_\theta=0\lra \forall i,j,\ \tr{\rho_\theta[L_{\theta,i},L_{\theta,j}]}=0$ 
can be shown by elementary algebra. 
Collecting these arguments proves Theorem \ref{thm3}. 
$\square$
\end{proof}

\subsubsection{Proof for Corollary \ref{cor1}}
\begin{proof} 
{\it Equivalence in condition \ref{cond_cl}:}\\
Since $G_\theta=\tilde{G}_\theta\,\lra\, G_\theta^{-1}=\tilde{G}_\theta^{-1}$, 
the first equivalence is immediate. 

To prove the second equivalence to $\tilde{G}_\theta^{-1}=\tilde{Z}_\theta$ in \ref{cond_cl}, 
let us assume first that a model is classical. 
Condition \ref{eq7} of Theorem \ref{thm4} gives
\be
\Dop{\tilde{L}_\theta^i}=0,\ \forall i.
\ee
Then, Eq.~\eqref{prop7} of Lemma \ref{lem5} yields $\tilde{g}_\theta^{ij}-\tilde{z}_\theta^{ij}=0$ 
for all $i,j$. Conversely, if $\tilde{G}_\theta^{-1}=\tilde{Z}_\theta$ holds, 
we have the following equivalence from the first matrix inequality in Lemma \ref{lem4-2}. 
\begin{align*}
\forall i,\,L_{\theta,i}=\tilde{L}_{\theta,i}&\lra G_\theta+\tilde{G}_\theta\tilde{Z}_\theta\tilde{G}_\theta= 2\tilde{G}_\theta\\
&\lra  G_\theta+\tilde{G}_\theta\tilde{G}_\theta^{-1}\tilde{G}_\theta= 2\tilde{G}_\theta\\
&\lra G_\theta=\tilde{G}_\theta. 
\end{align*}
This proves the converse part. 

The last equivalence to $Z_\theta=\tilde{Z}_\theta$ in \ref{cond_cl} is proven as follows. 
A classical model gives this condition is straightforward. Conversely, if this condition is satisfied, 
the second matrix inequality of Lemma \ref{lem4} is then expressed as
\be
Z_\theta\ge G_\theta^{-1}. 
\ee
Noting $G_\theta^{-1}=\Re Z_\theta$, this inequality concludes 
$\Im Z_\theta=0$. (Otherwise, the matrix inequality does not hold.) 
This then shows that the model is asymptotically classical, and we have 
$G_\theta^{-1}= Z_\theta=\tilde{Z}_\theta$. The condition $G_\theta^{-1}=\tilde{Z}_\theta$ 
holds if and only if the model is D-invariant from Lemma \ref{lem4}. 
Therefore, the model is asymptotically classical and D-invariant, i.e., the classical model. 
$\square$

\noindent
{\it Equivalence in condition \ref{cond_Dinv}:}\\
The first equivalence is already proven in Theorem \ref{thm3}, whose proof is given in Ref.~\cite{jsJMP}. 
Here we note that both conditions can be proven immediately if 
we use Lemma \ref{lem4}. 
$\square$

\noindent
{\it Equivalence in condition \ref{cond_ac}:}\\
This is proven in Theorem \ref{thm3}. 
$\square$
\end{proof}

\section{Examples}\label{sec5}
\subsection{Qubit models}
When the dimension of the Hilbert space is two, i.e., a qubit system, 
we can explicitly work out classification of models. 
To analyze a given qubit model, it is convenient to use the Bloch vector 
representation of qubit states. 
Define a three dimensional real vector $\v{s}_\theta=(s_\theta^i)$ for $i=1,2,3$ by
\be
s_\theta^i:=\tr{\rho_\theta\sigma_i},
\ee
where $\sigma_i$ are the standard Pauli matrices. 
Since the mapping $\v{s}_\theta\mapsto\rho_\theta$ is bijective, 
a quantum statistical model for the qubit case can be defined as
\be \label{qmodel_bloch}
\cM=\{\v{s}_\theta\,|\, \theta\in\Theta\}. 
\ee
Based on the Bloch vector $\v{s}_\theta$, we can derive closed formulas 
for the quantum score functions (SLD and RLD logarithmic derivative operators) 
and the quantum Fisher information matrices. (See, for example, Ref.~\cite{jsJMP}.) 
In Ref.~\cite{jsJMP}, we derived the following conditions for a given model \ref{qmodel_bloch} 
to be the D-invariant and asymptotically classical.
\begin{enumerate}
\item $\cM$ is D-invariant. \\
$\lra$ $|\v{s}_\theta|$ is independent of $\theta$. 
\item $\cM$ is asymptotically classical. \\ 
$\lra$ $\del_i \v{s}_\theta\times \del_j \v{s}_\theta$ ($\forall i\neq j$) is orthogonal to $\v{s}_\theta$. 
\end{enumerate}
The equivalent condition for the D-invariant model immediately tells us 
that any unitary model on the qubit system is D-invariant. 
The converse statement is, of course, not true in general. 
For example, the following two-parameter qubit model is D-invariant, but not unitary. 
\be
\cM=\{\v{s}_\theta=(\theta^1,\theta^2,\sqrt{s_0^2-(\theta^1)^2-(\theta^2)^2})\,|\,\theta\in\Theta\}, 
\ee
where $s_0\in(0,1)$ is a fixed constant and the parameters takes values 
within the region $\Theta\subset\bbr^2$ satisfying the positivity condition for the state.  

Next, we can work out whether or not there exists a classical qubit model. 
It is straightforward to show that there cannot be any multi-parameter classical qubit model under the regularity condition, 
and thus only one-parameter classical model exists. 
The reason is simply because there can be a single parameter classical model 
embedded in a $2\times 2$ matrix space. 
Any multi-parameter classical model becomes a non-regular model. 

Finally, we ask if there can be a quasi-classical model in a qubit system. 
It turns out that there exists no such a quasi-classical qubit model. 
This is due to the fact that imposing commutativity between the SLD operators 
leads to a non-regular model. 

To prove this statement, we note the commutation condition for the SLD operators  
is expressed in terms of the Bloch vectors as 
\be
[L_{\theta,i}\,,\,L_{\theta,j}]=0\ \lra\ \del_i \v{s}_\theta\times \del_j \v{s}_\theta=0. 
\ee
Consider a two-parameter qubit model. 
The condition $\del_1 \v{s}_\theta\times \del_2 \v{s}_\theta=0$ is 
equivalent to linearly dependence of two vectors 
$\del_1 \v{s}_\theta$, $\del_2 \v{s}_\theta$. 
This then implies the existence of a function $c:\Theta\to\bbr$ such that 
$L_{\theta,1}=c(\theta)L_{\theta,2}$ holds. 
This contradicts with linearly independence of the tangent vectors. 
Note that, if this is the case, the dimension of the tangent space is one rather than two. 
The case of three parameter models can be checked similarly. 

\subsection{Non-classical quasi-classical model} \label{sec-non_cl}
As we mentioned earlier, there exists a quantum statistical model 
that is quasi-classical (all SLD operators commute with each other) and non-classical. 
It is straightforward to observe that such cases arise if a model is non-regular. 
For example, quantum states are not full rank. Below, we give a simple 
regular statistical model in a three-dimensional quantum system ($d=3$). 

We consider the following two-parameter model:\\
$\cM:=\{\rho_\theta\,|\,\theta=(\theta^1,\theta^2)\in\Theta\}$, 
where 
\begin{align}
\rho_\theta&:=U_{\theta^2} \Lambda_{\theta^1} U_{\theta^2}^{-1},\\
\Lambda_{\theta^1}&:=
\left(\begin{array}{ccc}\lambda(\theta^1) & 0 & 0 \\0 & c \lambda(\theta^1) & 0 \\0 & 0 & 1-(1+c)\lambda(\theta^1)\end{array}\right) ,\\
U_{\theta^2}&:=\Exp{\I\theta^2 \sigma_1}\mbox{ with }\sigma_1=\left(\begin{array}{ccc}0 & 1 & 0 \\1 & 0 & 0 \\0 & 0 & 0\end{array}\right) , 
\end{align}
where a constant $c\in\bbr$ ($c\neq1$) and smooth function $\lambda(\theta^1)$ are chosen 
arbitrary as long as the corresponding classical model for $\Lambda_{\theta^1}$ 
\[
\cM_1:=\{p_{\theta^1}=(\lambda(\theta^1),c \lambda(\theta^1),1-(1+c)\lambda(\theta^1) )|\theta^1\in\Theta_1\},
\]
satisfies $\cM_1\in\cP(3)$. 
The SLD operators are calculated as
\begin{align}
L_{\theta,1}&=
U_{\theta^2}\,\frac{\dot{\lambda}}{\lambda}\left(\begin{array}{ccc}1 &0 & 0 \\ 0& 1 & 0 \\0 & 0 & -m(\theta^1)\end{array}\right)U_{\theta^2}^{-1}\\
L_{\theta,2}&=U_{\theta^2}\,2\frac{1-c}{1+c}\left(\begin{array}{ccc}0 & -\I & 0 \\ \I & 0 & 0 \\0 & 0 & 0\end{array}\right)U_{\theta^2}^{-1}, 
\end{align}
where $\dot{\lambda}=d\lambda(\theta^1)/d\theta^1$, $m(\theta^1)=1-(1-(1+c)\lambda(\theta^1))^{-1}$. 
To have a regular quantum model, we also impose $\dot{\lambda}\neq0$ for all $\theta^1$. 
It is clear that two SLD operators commute with each others for all $\theta$. 
The RLD operators are 
\begin{align}
\tilde{L}_{\theta,1}&=L_{\theta,1}\\
\tilde{L}_{\theta,2}&=U_{\theta^2}\,\left(\begin{array}{ccc}0 & -\I (1-c) & 0 \\ \I (1-c)/c & 0 & 0 \\0 & 0 & 0\end{array}\right)U_{\theta^2}^{-1}. 
\end{align}

We can show that the SLD Fisher information matrix is diagonal and is given by
\begin{align}
g_{\theta,11}&=\frac{\dot{\lambda}}{\lambda}(2+m(\theta^1)^2),\\
g_{\theta,12}&=g_{\theta,21}=0,\\
g_{\theta,22}&=\left(2\,\frac{1-c}{1+c}\right)^2\lambda(\theta^1). 
\end{align}
Whereas the RLD Fisher information matrix is 
\begin{align}
\tilde{g}_{\theta,11}&=g_{\theta,11},\\
\tilde{g}_{\theta,12}&=\tilde{g}_{\theta,21}=0,\\
\tilde{g}_{\theta,22}&=\frac{(1-c)^2(1+c)}{c}\lambda(\theta^1). 
\end{align}
It is easy to see that $\tilde{g}_{\theta,22} \ge g_{\theta,22}$ with equality 
if and only if $c=1$, which is excluded. 
Therefore, $G_\theta\neq\tilde{G}_\theta$ holds and this model is not classical by Theorem \ref{thm4}. 

\section{Concluding remarks}\label{sec6}
We have derived classification and several equivalent characterizations of quantum parametric models 
based on the estimation error bound, the Holevo bound. 
Three classes are mainly discussed in this paper: the classical model, D-invariant model, and asymptotically classical model.  
We have also given relationships among these classes. In particular, the classical model 
can be viewed as the intersection of the D-invariant and asymptotically classical models. 
The classical model has several different interpretations based on the geometrical point of view. 
Although all conditions are related to on another, they show different side-sights on the classical model 
as a sub-model of the general quantum statistical model. 
We have also analyzed the quasi-classical model, in which all SLD operators commute with each other, 
and have shown that it is still quantum model. 

Before closing the paper, we list several open questions to be addressed. 
In this paper, we have focused on the global aspects of the quantum statistical models only. 
First extension is to analyze local properties of each class of the quantum model. 
This is important to understand their properties from the point of view of information geometry. 
In Ref.~\cite{jsJMP}, we analyzed the local properties for the D-invariant and asymptotically classical models.  
Therefore, it is interesting to see where the local classical model is a useful concept or not. 
Second, we don't know how much local properties determine the global property for a given model. 
An interesting question is then to ask whether we can characterize the class model globally by local conditions. 
Third, we have only used two different quantum Fisher information matrices, 
the SLD and RLD Firsher, together with their dual matrices $Z_\theta$ and $\tilde{Z}_\theta$. 
We expect that other families of quantum Fisher information should also give model classification and characterization. 
Last, there should other important classes for the quantum parametric models other than discussed in this paper. 
These are some of untouched questions in this paper, and should be examined in the subsequent publication. 

\section*{Acknowledgment}
The authors would like to thank Dr. S.~Ragy for providing information about 
their results prior to a publication \cite{RJDD16}. 
The work is partly supported by JSPS KAKENHI Grant Number JP17K05571.


\begin{thebibliography}{20}
\bibitem{comment}
We note that the previous publication contains a typo in the formula. 
Lemma \ref{lem5} reports the corrected version.
\end{thebibliography}
\end{document}